\documentclass[letterpaper, 10 pt, conference]{ieeeconf}  

\IEEEoverridecommandlockouts                              

\usepackage{graphics} 
\usepackage{epsfig} 
\usepackage{mathptmx} 
\usepackage{times} 
\usepackage{amsmath} 
\usepackage{amssymb}  

\usepackage{graphics,graphicx,psfrag,array,eepic}
\usepackage{multirow}
\usepackage{hhline}
\usepackage{mathtools}
\usepackage[strict]{changepage}
\usepackage{epsfig}
\usepackage{amsmath}
\usepackage{amssymb}

\usepackage{float}
\restylefloat{table}

\usepackage{multicol}
\usepackage{color}
\usepackage{balance}
\usepackage{cite}
\usepackage{booktabs} 
\usepackage{caption,booktabs}
\usepackage[tableposition=top]{caption}
\usepackage{subcaption}

\usepackage{lipsum}
\usepackage{comment}
\usepackage{amssymb}
\usepackage{graphicx}
\usepackage{algorithm}
\usepackage{algorithmic}
\usepackage{comment}

\usepackage{mathtools, cuted}
\usepackage[thinc]{esdiff}

\usepackage{psfrag}
\usepackage{array}
\usepackage{latexsym,theorem}
\usepackage{amsfonts,amssymb,amsbsy}
\usepackage{tikz}
\usetikzlibrary{decorations.shapes,shapes.geometric,shadows,arrows,automata,positioning,calendar,mindmap,backgrounds,scopes,chains,er,3d,matrix,fit}
\usepackage{tikz-cd}
\usepackage{pgfplots}
\usetikzlibrary{intersections, pgfplots.fillbetween}

\newcommand{\R}{\mathbb{R}}

\newcommand{\diag}{\operatorname{diag}}

\newtheorem{theorem}{Theorem}

\newtheorem{assumption}{Assumption}
\newtheorem{remark}{Remark}

\usepackage{color}

\title{\LARGE \bf
AIMD scheduling and resource allocation in distributed \\ computing systems}
\author{Eleftherios Vlahakis, Nikolaos Athanasopoulos, Se\'{a}n McLoone
\thanks{*This work is supported by the CHIST-ERA 2018 project DRUID-NET
``Edge Computing Resource Allocation for Dynamic Networks'' \texttt{https://druidnet.netmode.ntua.gr/}}
\thanks{Authors are with the School of Electronics, Electrical Engineering and Computer Science, Queen’s University Belfast, Northern Ireland, UK. E-mail addresses: {\tt\small \{e.vlahakis, n.athanasopoulos, s.mcloone\}@qub.ac.uk}}%
}

\begin{document}

\maketitle
\thispagestyle{empty}
\pagestyle{empty}

\begin{abstract}
We consider the problem of simultaneous scheduling and resource allocation of an incoming flow of requests to a set of computing units. By representing each computing unit as a node, we model the overall system as a multi-queue scheme. Inspired by congestion control approaches in communication networks, we propose an AIMD-like (additive increase multiplicative decrease) admission control policy that is stable irrespective of the total number of nodes and AIMD parameters. The admission policy allows us to establish an event-driven discrete model, triggered by a locally identifiable enabling condition. Subsequently, we propose a decentralized resource allocation strategy via a simple nonlinear state feedback controller, guaranteeing global convergence to a bounded set in finite time. Last, we reveal the connection of these properties with Quality of Service specifications, by calculating local queuing time via a simple formula consistent with Little's Law. 
\end{abstract}


\section{Introduction}









Distributed computing is a new paradigm emerging to address the growing demand for extensive, real-time computations at the \textit{edge} as a result of the growing number of end-users (e.g., smart devices, sensors) connected to the edge of the Internet. Although this emerging technology opens new opportunities for more sophisticated applications (see, e.g., \cite{Kehoe2015,Mach2017,Abbas2018}), it presents several research challenges, especially in the context of resource allocation and control of edge-servers due to factors such as the need to take account of latency constraints, limited capacity of edge-servers, and its inherent decentralized structure.

Feedback control has been a powerful mathematical tool for tackling management problems in the context of modern computer systems \cite{Karamanolis2005}. Given a representative dynamical model, control theory allows analytical derivation of formal guarantees and certificates. However, modelling computer systems is a formidable task by itself, thus, many works rely on application-specific models obtained via system identification methods. See for example \cite{Wang2005,Dechouniotis2015,Avgeris2019}. Focusing on a more abstract modelling paradigm agnostic to each individual node specificities, we follow a queueing system modelling approach that enhances scalability, naturally, at the expense of accuracy loss. Notable works avoiding application-specific modelling can be found in \cite{Maggio2010,Kalyvianaki2014,Makridis2018}.

The control problem considered in this paper consists of $1)$ the scheduling of a stream of requests, and $2)$ the resource allocation of a set of computing units associated with a specific application. Representing each computing unit as a node and associating each node with a queue, we model the entire scheme as a multi-queue system. We assume that there is no interaction between nodes and computing units are independent from each other. A central node acts as an aggregation point, receiving all requests and dispatching them to individual nodes. Queues in this work are consistent with the \textit{First Come First Served} (FCFS) selection policy. 


Our approach to scheduling and resource allocation is motivated by the Additive Increase Multiplicative Decrease (AIMD) algorithm, a celebrated method in network management. The AIMD algorithm was originally introduced in \cite{Chiu1989} for tackling congestion phenomena in computer networks in a robust and decentralized manner requiring minimum interaction between nodes. Since then, it has become a fundamental building block of the Transmission Control Protocol (TCP) widely used across the Internet. An excellent and comprehensive study of the AIMD algorithm with several extensions and applications can be found in \cite{Corless2016}. 

In this paper, we study how an AIMD-inspired simple admission control policy can be utilized for general scheduling problems. A typical AIMD model results in an event-driven discrete controller which is triggered by a capacity event associated with constraints, e.g., bandwidth constraints. Berman \textit{et al.} in \cite{Berman2004} and Shorten \textit{et al.} in \cite{Shorten2005} show that such a control scheme can be formulated as positive system, thus, stability and convergence properties can be derived from the Perron–Frobenius Theorem. A first challenge we face is that a positive system formulation is not possible in our case due to the absence of a capacity constraint in a scheduling task. Instead, we consider a queue clearance event and manage to show stability via a significant result in Linear Algebra (cf. \cite{Golub1973,Bunch1978}) involving the eigenproblem of rank-one perturbations of symmetric matrices (Theorems \ref{thm:rankOnePerturb} and \ref{thm:PhiStability}). This formulation leads to a new admission control algorithm with AIMD structure which is stable irrespective of the AIMD tuning and the number of nodes, and inherits attractive features of the standard AIMD algorithm (e.g., fairness among nodes, tunable convergence rate) \cite{Shorten2006}. To the best of our knowledge, this paper presents a new admission control policy with AIMD dynamics for scheduling tasks.

As a result of the simplicity of the AIMD scheduling policy proposed, we formulate a resource allocation strategy defined as a decentralized globally stabilizing nonlinear feedback controller. Following a set-theoretic approach, 
we show that under the proposed resource allocation law, individual queues are bounded, and, further, converge in finite time to a well-defined interval which is invariant \cite{Blanchini2015}. This effectively permits a priori analysis of Quality of Service (QoS) metrics, such as queueing time. Overall, scheduling and resource allocation lie in the same control loop leading to a simple decentralized system which is stable, scalable, and locally configurable.

Unlike standard stochastic methods, see, e.g.,  \cite{Maguluri2013,Maguluri2014}), we follow a deterministic approach to workload modelling. This choice simplifies the simultaneous scheduling and resource allocation problem, and most importantly, leads to deterministic performance certificates.  A relaxation of our results towards a non-deterministic workload as well as the incorporation of constraints will be considered in future work. 

The remainder of the paper is organized as follows. The notation used in the paper is introduced in Section \ref{section:notation}, while underpinning definitions and assumptions are given in Section \ref{section:definitions}. The main results of the paper, namely, the AIMD scheduling strategy, the resource allocation control, and the calculation of queueing time are then presented in Sections \ref{section:AIMD}, \ref{section:resourceAllocation}, and \ref{section:queueingTime}, respectively. In Section \ref{section:example} we highlight our results via an illustrative numerical example. Finally, Section \ref{section:conclusion} discusses our main results and future research directions.


\section{Notation}\label{section:notation}

The field of real numbers is denoted by $\mathbb{R}$. $\mathbb{R}^{n}$ denotes the $n$-dimensional vector space over the field $\mathbb{R}$, and $\mathbb{R}^{n\times m}$ denotes the set of $n\times m$ real matrices. The transpose of $\xi$ is denoted by $\xi'$. Let $x_{1},\;\ldots,\;x_{n}$ be vectors not necessarily of the same dimensions. Then, $\hat{x}=\mathrm{Col}(x_{1},\;\ldots,\;x_{n}) = (x_{1}'\;\cdots\;x_{n}')'$. Let $a_{1},\;\ldots,\;a_{n}\in \mathbb{R}$, then, $a = (a_{1},\ldots,a_{n})\in \mathbb{R}^{n}$, and $A = \mathrm{diag}(a_{1},\;\ldots,\;a_{n})$ is a diagonal matrix, with $a_{1},\;\ldots,\;a_{n}$ as its diagonal entries. We denote by $\det(A)$ the determinant of a square matrix $A$. The identity matrix of dimension $m\times m$ is denoted by $I_{m}\in\mathbb{R}^{m\times m}$ unless the dimensions are obvious in which case the subscript will be omitted. Matrix $\Xi\in\mathbb{R}^{n\times n}$ is called symmetric if $\Xi' = \Xi$. Let $\lambda_{i}(\Phi)$ be the $i$th eigenvalue of matrix $\Phi\in\mathbb{R}^{m\times m}$, with $i=1,\ldots,m$. Then, the spectrum of $\Phi$ is denoted by $\sigma(\Phi) = \{\lambda_{1}(\Phi),\ldots,\lambda_{m}(\Phi)\}$. A matrix $\Phi\in\mathbb{R}^{m\times m}$ is  Schur if all its eigenvalues strictly lie inside the unit circle, i.e., $|\lambda_{i}(\Phi)|<1$, $i=1,\ldots,m$.

\section{Definitions and Basic Assumptions}\label{section:definitions}

\subsection{Single-queue system}

We define a \textit{request} as an individual demand for computing resources provided by a computing node. A \textit{computing node} is defined as the physical (or virtual) computing environment, consisting of hardware, software, and network resources, whereby a request is executed. A \textit{queue} is defined as the waiting mechanism whereby a request arriving at a node is temporarily put on hold until it is selected for service from among other requests that are waiting. Here, we consider queues consistent with the \textit{First Come First Served} (FCFS) selection principle. 

A \textit{queueing system} \cite{Kleinrock1975,Cassandras2008} is defined as the dynamic relationship that is developed between a stream of request arrivals at and a flow of request departures from a computing device, respectively, in the presence of a queue. From a mathematical perspective, a queue acts as an integrator of the difference between arrival and departure rates. A simple queueing system is depicted in Fig. \ref{fig:single-queue} which is consistent with the following notation. 

Denoting by $\mathbf{X}(t)$ and $\mathbf{Y}(t)$ the arrivals at and departures from a queue, respectively, in interval $[0,\;t]$, the number of queued requests at time $t$, $\mathbf{Q}(t)$, is defined as the difference between arrivals and departures in interval $[0,\;t]$. Note that exact knowledge of $\mathbf{X}(t)$ and $\mathbf{Y}(t)$ is typically impossible in real applications, with arrivals and departures considered as stochastic processes described by appropriate probability distributions. A comprehensive overview of stochastic queueing systems can be found in \cite{Kleinrock1975}. Here, to highlight the admission and resource allocation control strategies proposed in the paper, we simplify our model structure following a deterministic approach. Specifically, we assume the following. 
\begin{assumption}\label{ass:1}
    \begin{itemize}
        \item[]
        \item[(A1)] The arrival rate, denoted by $\lambda(t) = \frac{\mathbf{d}}{\mathbf{d}t}\mathbf{X}(t)$, is constant. 
        \item[(A2)] Requests arriving at a queueing system are identical in terms of the combination of computing resources (CPU time, memory, disk space) required to serve them.
        \end{itemize}
\end{assumption}
%
%
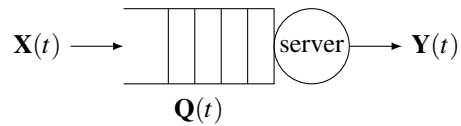
\begin{figure}[ht]
\centering
\scalebox{1}{
\begin{tikzpicture}[>=latex]
\draw (0,0) -- ++(2cm,0) -- ++(0,-1.cm) -- ++(-2cm,0);
\foreach \i in {1,...,4}
  \draw (2cm-\i*10pt,0) -- +(0,-1.cm);

\draw (2.5,-0.5cm) circle [radius=0.5cm] node[] {server};

\draw[->] (3.,-0.5) -- +(20pt,0) node[right] {$\mathbf{Y}(t)$};
\draw[<-] (0,-0.5) -- +(-20pt,0) node[left] {$\mathbf{X}(t)$};
\node at (1cm,-1.35cm) {$\mathbf{Q}(t)$};
\end{tikzpicture}}
\caption{A simple queueing system.
} 
\label{fig:single-queue}
\end{figure}

\subsection{Queueing time}

\textit{Queueing time} (also termed waiting time or latency \cite{Kleinrock1975}) is the main performance metric of a queueing system (e.g, in edge computing applications), expressing the time that a request is expected to be queued before processed. Given the knowledge of arrivals at and departures from a queue in interval $[0,\;t]$, 
we define queueing time as
\begin{equation}\label{eq:queueingTime}
    T_{q}(t) =\frac{\int_{0}^{t}\mathbf{Q}(s)\mathbf{d}s}{\mathbf{X}(t)},
\end{equation}
where $\mathbf{Q}(t) =  \mathbf{X}(t) - \mathbf{Y}(t)$. Note that the integral in the numerator on the right side of \eqref{eq:queueingTime} expresses the aggregate queueing time of all requests arriving in $[0,\;t]$ measured in $\textnormal{requests}\times\textnormal{seconds}$. By averaging the aggregate queueing time and the request arrivals, respectively, by the length of interval $[0,\;t]$ as $T_{q}(t) =\frac{\int_{0}^{t}\mathbf{Q}(s)\mathbf{d}s/t}{\mathbf{X}(t)/t}$, it is easy to see that definition \eqref{eq:queueingTime}  is in agreement with \textit{Little's Law}, which states that the average number of queued requests $\int_{0}^{t}\mathbf{Q}(s)\mathbf{d}s/t$ is equal to the product of average arrival rate $\mathbf{X}(t)/t$ and queueing time $T_{q}(t)$. A detailed description of \textit{Little's Law} can be found in \cite[Chapter~2]{Kleinrock1975}.




\subsection{Event-driven discretization and event generator}

An event generator is introduced as the mechanism indicating time instants at which a well-defined (triggering) condition $\mathcal{C}_{\varepsilon}$ is satisfied. Condition satisfaction can be written as
\begin{equation}
    \mathcal{C}_{\varepsilon}(t_{k}) = \texttt{true},
\end{equation}
where $t_{k}$ denotes the time instant at which the $k$th event occurs (is generated). Note that time events can be modelled by casting the continuous time as an autonomous state variable, namely,
\begin{equation}
    t(k+1) = t(k) + T(k),
\end{equation}
where $T(k)$ is the time-varying sampling period.


We emphasize that the facilitation of an aperiodic model (with respect to time) derivation will be the result of two main design strategies, namely, 
\begin{itemize}
    \item[$1)$] the introduction of a \textit{batch} queue into the system,
    \item[$2)$] the adoption of an AIMD admission control policy.
\end{itemize}
This strategic choice is now exemplified via a simple tandem queueing system.

\subsection{Tandem queueing system with AIMD dynamics}

We consider the two-queue system (also termed tandem queueing system) shown in Fig. \ref{fig:tandem1}, where $\lambda(t)$ is a piece-wise differentiable function representing workload, $\delta(t)$ is the number of queued requests waiting at queue $Q_{1}$ to be dispatched to queue $Q_{2}$ at an admission rate $u(t)$, while $w(t)$ and $\gamma(t)$ represent queued requests and service rate, respectively. 
\begin{figure}[ht]
\centering
\scalebox{1}{
\begin{tikzpicture}[>=latex]
\draw (0,0) -- ++(2cm,0) -- ++(0,-.75cm) -- ++(-2cm,0);
\foreach \i in {1,...,4}
  \draw (2cm-\i*10pt,0) -- +(0,-.75cm);
\draw (2.375,-0.375cm) circle [radius=0.375cm] node[] {$u(t)$};
\draw[<-] (0,-0.375) -- +(-20pt,0) node[left] {$\lambda(t)$};
\node at (1cm, .25cm) {$Q_{1}$};
\node at (1cm,-1.05cm) {$\delta(t)$};
\begin{scope}[xshift = 3.5cm]
\draw (0,0) -- ++(2cm,0) -- ++(0,-.75cm) -- ++(-2cm,0);
\foreach \i in {1,...,4}
  \draw (2cm-\i*10pt,0) -- +(0,-.75cm);
\draw (2.375,-0.375cm) circle [radius=0.375cm] node[] {$\gamma(t)$};
\draw[<-] (0,-0.375) -- +(-20pt,0);
\node at (1cm, .25cm) {$Q_{2}$};
\node at (1cm,-1.05cm) {$w(t)$};
\end{scope}
\end{tikzpicture}}
\caption{A tandem queueing system.} 
\label{fig:tandem1}
\end{figure}
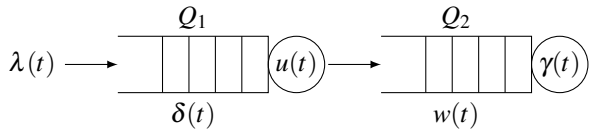
Using this notation, the continuous-time dynamics of the two-queue system can be written in a compact form as
\begin{align}\label{eq:continousSystem_delta_w}
\begin{bmatrix} \dot{\delta}(t) \\ \dot{w}(t) \end{bmatrix} = \begin{bmatrix}
     1 & -1 & 0 \\ 0 & 1 & -1 
\end{bmatrix} \begin{bmatrix}
     \lambda(t) \\ u(t) \\ \gamma(t)
\end{bmatrix}.
\end{align}

Before proceeding with the discretization of the model, we define a triggering condition that enables the generation of an event indicating the commencement of a new cycle. Let $u(t)$ be an admission control policy such that 
\begin{equation}\label{eq:condition_delta}
    \delta(t_{k}) = 0,
\end{equation}
i.e., all the requests that have arrived at queue $Q_{1}$ by time $t_{k}$ have been admitted to queue $Q_{2}$. Hence, in this regard, $Q_{1}$ instantaneously becomes empty at $t_{k}$. To ensure that condition \eqref{eq:condition_delta} can always be satisfied at a finite time for a constant $\lambda(t)$, we design $u(t)$ as an AIMD controller as follows. Let 
\begin{align}
    u(t_{k}^{-}) = \lim_{\substack{t\to t_{k} \\ t< t_{k}}} u(t), \;\;
    u(t_{k}^{+}) = \lim_{\substack{t\to t_{k} \\ t> t_{k}}} u(t),
\end{align}
where $t_{k}^{-}$ is the ending time of the $(k-1)$th cycle, and $t_{k}^{+}$ the starting time of the $k$th cycle. Since $\delta(t_{k}^{+}) = 0$, we let
\begin{equation}\label{eq:backoffing}
    u(t_{k}^{+}) = \beta u(t_{k}^{-}),
\end{equation}
where $0 < \beta < 1$ is called the \textit{backoff parameter}. While queue $Q_{1}$ remains empty, admission control $u(t^{+})$ shrinks to a fraction of $u(t^{-})$ according to \eqref{eq:backoffing}. This is called the \textit{Multiplicative Decrease} (MD) phase of the cycle. By the time queue $Q_{1}$ starts growing, admission rate $u(t)$ increases in a ramp fashion as
\begin{equation}\label{eq:continuous_u(t)}
    u(t) = \beta u(t_{k}^{-}) + \alpha (t-t_{k}), \; t\geq t_{k},
\end{equation}
where the slope of the ramp $\alpha >0$ is called the \textit{growth rate}. Let now $t_{k+1}^{-}$ be the ending time of the $k$th cycle, i.e., $\delta(t_{k+1}) = 0$. Then, the duration of the $k$th cycle is called the \textit{cycle period} and is denoted by $T(k) = t_{k+1} - t_{k}$. Similarly, the interval $(t_{k}^{+},\;t_{k+1}^{-})$ is called the \textit{Additive Increase} phase. Denoting time instants by $k=t_k$, with 
$k\geq 0$, an event-driven discrete model is derived as
\begin{equation}
    \begin{bmatrix}
         w(k+1) \\ u(k+1) \\ \delta(k+1)
    \end{bmatrix} =
    \begin{bmatrix}
         w(k) + (\beta u(k) + \frac{\alpha}{2}T(k) - \gamma(k))T(k) \\
         \beta u(k) + \alpha T(k) \\ \delta(k)
    \end{bmatrix},
\end{equation}
where $u(k)$ is an AIMD controller with triggering condition
\begin{equation}
    \delta (k) = 0.
\end{equation}
Next, we generalize the AIMD admission control approach to a system with multiple queues, and examine the properties of the AIMD algorithm and its effect on the entire system dynamics.

\subsection{Multi-queue system}

We consider a set of $n$ computing nodes
represented by a multi-queue system, where each node is modelled by a queue combined with a (physical or virtual) computing environment. We assume that a constant workload $\lambda$ enters the system via a batch queue, which is independent of the computing nodes. The workload is manifested as a flow of requests that are dispatched to $n$ nodes according to an admission control policy $u_{i}(t)$, $i=1,\;\ldots,\;n$, with each node representing a computing unit. We denote the number of queued requests that have not yet been admitted at time $t$ by $\delta(t)$, and the number of admitted requests waiting to be selected for service by the $i$th node at time $t$ by $w_{i}(t)$. Service rate of the $i$th node is denoted by $\gamma_{i}$, $i=1,\;\ldots,\;n$. The entire system is depicted in Fig. \ref{fig:tandem_multi_2}. Next, we examine the aggregate dynamics of a large-scale AIMD admission controller.
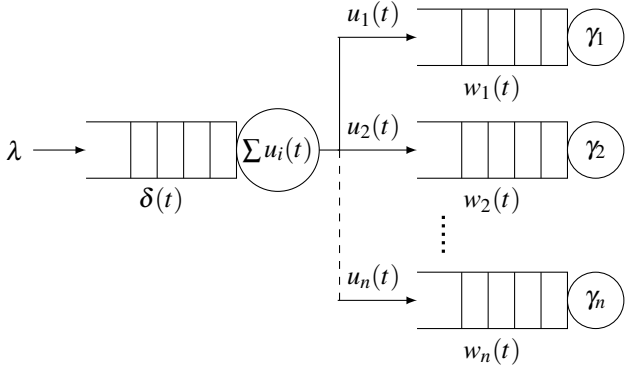
\begin{figure}[ht]
\centering
\scalebox{1}{
\begin{tikzpicture}[>=latex]
\draw (0,0) -- ++(2cm,0) -- ++(0,-.75cm) -- ++(-2cm,0);
\foreach \i in {1,...,4}
  \draw (2cm-\i*10pt,0) -- +(0,-.75cm);
\draw (2.55,-0.375cm) circle [radius=0.55cm] node[] {$\sum u_{i}(t)$};
\draw[-] (3.1, -.375cm) -- (3.37, -.375cm) -- (3.37, -.375cm +1.5cm);
\draw[dashed]  (3.37, -.375cm) -- (3.37, -.375cm -2cm);
\draw[very thick, dotted]  (4.75, -1.25cm) -- (4.75, -1.75cm);
\draw[<-] (0,-0.375) -- +(-20pt,0) node[left] {$\lambda$};
\node at (1cm,-1.05cm) {$\delta(t)$};
\begin{scope}[xshift = 4.4cm, yshift = 1.5cm]
\draw (0,0) -- ++(2cm,0) -- ++(0,-.75cm) -- ++(-2cm,0);
\foreach \i in {1,...,4}
  \draw (2cm-\i*10pt,0) -- +(0,-.75cm);
\draw (2.375,-0.375cm) circle [radius=0.375cm] node[] {$\gamma_{1}$};
\draw[<-] (0.,-0.375) -- +(-30pt,0) node[anchor = south west] {$u_{1}(t)$};
\node at (1cm,-1.05cm) {$w_{1}(t)$};
\end{scope}
\begin{scope}[xshift = 4.4cm, yshift = 0cm]
\draw (0,0) -- ++(2cm,0) -- ++(0,-.75cm) -- ++(-2cm,0);
\foreach \i in {1,...,4}
  \draw (2cm-\i*10pt,0) -- +(0,-.75cm);
\draw (2.375,-0.375cm) circle [radius=0.375cm] node[] {$\gamma_{2}$};
\draw[<-] (0.,-0.375) -- +(-30pt,0) node[anchor = south west] {$u_{2}(t)$};
\node at (1cm,-1.05cm) {$w_{2}(t)$};
\end{scope}
\begin{scope}[xshift = 4.4cm, yshift = -2.cm]
\draw (0,0) -- ++(2cm,0) -- ++(0,-.75cm) -- ++(-2cm,0);
\foreach \i in {1,...,4}
  \draw (2cm-\i*10pt,0) -- +(0,-.75cm);
\draw (2.375,-0.375cm) circle [radius=0.375cm] node[] {$\gamma_{n}$};
\draw[<-] (0.,-0.375) -- +(-30pt,0) node[anchor = south west] {$u_{n}(t)$};
\node at (1cm,-1.05cm) {$w_{n}(t)$};
\end{scope}
\end{tikzpicture}}
\caption{A multi-queue system with AIMD admission control policy.} 
\label{fig:tandem_multi_2}
\end{figure}

\section{AIMD Admission Control}\label{section:AIMD}


We consider the system of $n$ computing nodes depicted in Fig. \ref{fig:tandem_multi_2}. The number of queued (unadmitted) requests at the beginning of the $(k+1)$th event is given by
\begin{equation}\label{eq:deltak+1}
    \delta(k+1) = \delta(k) + \lambda T(k) - \int_{t_{k}}^{t_{k+1}} \sum u_{i}(t) \mathbf{d}t.
\end{equation}
We recall that at each event $k, \; k+1,\;\ldots$, we have
\begin{equation}\label{eq:triggeringcondition}
    \delta(k) = \delta(k+1) = \cdots = 0.
\end{equation}
The AIMD formulation of the admission controller  yields an exact formula for the cycle period $T(k)$ permitting a closed form of the aggregate admission control system. To this purpose, during the AI phase, the $i$th admission rate ramps up as follows,
\begin{equation}\label{eq:u_i(t)}
    u_{i}(t) = \beta_{i}u_{i}(t_{k}) + \alpha_{i}(t-t_{k}), \; i = 1,\ldots,n,
\end{equation}
which is a continuous-time controller for $t \in [t_{k},\;t_{k+1})$. Based on condition \eqref{eq:triggeringcondition}, the event-based dynamics of the $i$th admission controller is written as:
\begin{equation}\label{eq:u_ik+1}
    u_{i}(k+1) = \beta_{i}u_{i}(k) + \alpha_{i}T(k), \; i = 1,\ldots,n.
\end{equation}
In view of the triggering condition \eqref{eq:triggeringcondition} and using \eqref{eq:u_ik+1} in \eqref{eq:deltak+1}, we get
\begin{equation}
    \lambda T(k) = \sum_{i=1}^{n} (2\beta_{i}u_{i}(k) + \alpha_{i}T(k)) \frac{T(k)}{2}\footnote{the integral on the right side of \eqref{eq:deltak+1} is the sum of areas of $n$ trapezoids due to the AIMD dynamics of $u_{i}$, i = 1,\ldots, n.},
\end{equation}
or
\begin{equation}\label{eq:lambda}
    \lambda = \sum_{i=1}^{n} (\beta_{i}u_{i}(k) + \frac{\alpha_{i}}{2}T(k)),
\end{equation}
from which, the cycle period is defined as
\begin{equation}\label{eq:eventPeriod}
    T(k) = \frac{\lambda - \sum_{i=1}^{n}\beta_{i} u_{i}(k) }{\sum_{i=1}^{n}\frac{\alpha_{i}}{2}}.
\end{equation}
From \eqref{eq:eventPeriod}, we may write that
\begin{equation}
    u_{i}(k+1) = \beta_{i}u_{i}(k) +\alpha_{i} \frac{\lambda - \sum_{i=1}^{n}\beta_{i} u_{i}(k) }{\sum_{i=1}^{n}\frac{\alpha_{i}}{2}}.
\end{equation}
Now, defining
\begin{align}
    U(k) &= \mathrm{Col}(u_{1}(k),\ldots,u_{n}(k)),\\
    \bar{\alpha} &= \frac{1}{\sum_{j=1}^{n}\alpha_{j}}\mathrm{Col}(\alpha_{1},\ldots,\alpha_{n}),\\
    B &= \diag(\beta_{1},\ldots,\beta_{n}),\\
    \beta & = (\beta_{1},\ldots,\beta_{n}),
\end{align}
the aggregate admission control system can be expressed as
\begin{equation}\label{eq:U(k+1)}
    U(k+1) = \Phi U(k) + 2\bar{\alpha}\lambda,
\end{equation}
where $\Phi = B - 2\bar{\alpha} \beta'$, with $\bar{\alpha}'\mathbf{1}=1$. Next, we show that system \eqref{eq:U(k+1)} is stable, thus, the $i$th AIMD admission controller $u_{i}(k)$ converges to a unique equilibrium point $u_{i}^{\ast}$. We first present the following result, which appears in several works in the context of Linear Algebra, see, e.g., \cite[Theorem~1]{Bunch1978}, \cite[Section~5]{Golub1973}.  
\begin{theorem}\label{thm:rankOnePerturb}
Let $C = D + \rho zz'$, where $D\in\mathbb{R}^{n\times n}$ is diagonal, $\rho\in\mathbb{R}$, and $z\in\mathbb{R}^{n}$. Let $d_{1}\leq d_{2}\leq \ldots \leq d_{n}$ be the eigenvalues of $D$, and $c_{1}\leq c_{2}\leq \ldots \leq c_{n}$ be the eigenvalues of $C$. Then,
\begin{itemize}
    \item[i.)] $d_{1}\leq c_{1} \leq d_{2}\leq c_{2} \leq \ldots \leq d_{n} \leq c_{n}$ if $\rho>0$,
    \item[ii.)] $c_{1}\leq d_{1} \leq c_{2}\leq d_{2} \leq \ldots \leq c_{n} \leq d_{n}$ if $\rho < 0$. 
\end{itemize}
If $d_{1},\ldots,d_{n}$ are distinct and all the elements of $z$ are nonzero, then $c_{1},\ldots,c_{n}$, namely the eigenvalues of $C$, strictly separate the eigenvalues of $D$.
\end{theorem}
We are in a position to state the first main result, namely, the stability of the AIMD scheduling policy.
\begin{theorem}\label{thm:PhiStability}
Let vectors $\alpha = (\alpha_{1},\;\ldots,\;\alpha_{n})$, $\beta = (\beta_{1},\;\ldots,\;\beta_{n})$, where $0\leq \alpha_{i}\leq 1$, $0 < \beta_{i} < 1$, $\forall i = 1,\;\ldots,\;n$, and $\mathbf{1}'\alpha = 1$, with $\mathbf{1}= (1,\;\ldots,\;1)$. Let also $B = \diag(\beta_{1},\ldots,\beta_{n})$. Then,
\begin{equation}
    \Phi = B - 2\alpha \beta',
\end{equation}
is a Schur matrix.
\end{theorem}
\begin{proof}
Matrix $\Phi$ can also be written as 
\begin{equation}
    \Phi = (I-2A)B,
\end{equation}
where $A = \alpha \mathbf{1}'$ is a rank-one matrix with $\sigma(A) = \{\mathbf{1}'\alpha, 0 ,\ldots, 0\}$, and $\mathbf{1}'\alpha = 1$ by definition. In the sequel, we denote by $\sigma(\Phi) = \{\phi_{1},\ldots,\phi_{n}\}$ the spectrum of $\Phi$. Clearly, $\sigma(B) = \{\beta_{1},\ldots,\beta_{n}\}$. Also, it is easy to show that $\sigma(I-2A) =\{-1,\;1,\ldots,1\}$. We can also write that $\det(\Phi) = \phi_{1}\phi_{2}\ldots\phi_{n}$, and $\det(\Phi) = \det(B)\det(I-2A)$. Thus,
\begin{equation}\label{eq:phi12nbeta12n}
    \phi_{1}\phi_{2}\cdots\phi_{n} = -\beta_{1}\beta_{2}\cdots\beta_{n}.
\end{equation}
Let now $\hat{A} = \diag(\alpha_{1},\ldots, \alpha_{n})$, and
\begin{equation}
    \hat{\Phi} = B^{\frac{1}{2}}\hat{A}^{-\frac{1}{2}}\Phi\hat{A}^{\frac{1}{2}}B^{-\frac{1}{2}}.
\end{equation}
Clearly, matrices $\Phi$ and $\hat{\Phi}$ are similar, and therefore have identical eigenvalues. Note also that $\hat{\Phi}$ can be written as
\begin{equation}
    \hat{\Phi} = B - 2 zz',
\end{equation}
which is clearly a symmetric matrix, where 
\begin{equation}
    z = \begin{bmatrix}
     \sqrt{\alpha_{1}\beta_{1}} & \sqrt{\alpha_{2}\beta_{2}} & \cdots & \sqrt{\alpha_{n}\beta_{n}}
    \end{bmatrix}.
\end{equation}

Without loss of generality, let $b_{1}\leq \ldots \leq b_{n}$, and $\phi_{1}\leq \ldots \leq \phi_{n}$. Then,  from Theorem \ref{thm:rankOnePerturb}, and since all elements of $z$ are nonzero, we may write that
\begin{equation}\label{eq:phileqbeta}
    \phi_{1} \leq \beta_{1} \leq \phi_{2} \leq \beta_{2} \leq \ldots \leq \phi_{n} \leq \beta_{n}.
\end{equation}
From \eqref{eq:phi12nbeta12n} and \eqref{eq:phileqbeta}, we can conclude that $0<\phi_{2},\;\phi_{3},\;\dots,\;\phi_{n}<1$, and $\phi_{1}$ is a negative real number. From \eqref{eq:phileqbeta}, we have that
\begin{equation}\label{eq:phi2phin}
    \phi_{2}\phi_{3}\cdots\phi_{n} \geq \beta_{1}\beta_{2}\cdots\beta_{n-1}.
\end{equation}
However, due to \eqref{eq:phi12nbeta12n}, \eqref{eq:phi2phin} implies that 
\begin{equation}
    |\phi_{1}| \leq \beta_{n},
\end{equation}
i.e., $-\beta_{n}\leq \phi_{1}<0$. Thus, all the eigenvalues of $\hat{\Phi}$ (hence $\Phi$), strictly lie inside the unit circle (specifically on the real axis between $-1$ and $1$). This proves the theorem.
\end{proof}

The main deductions that follow from the analysis presented in this section are as follows:
\begin{itemize}
    \item[1)] AIMD parameters $\alpha_{i}$, $\beta_{i}$, $i=1,\ldots,n$, can be locally selected at each individual node. Thus, system \eqref{eq:u_ik+1} represents a decentralized admission control policy.
    \item[2)] In view of Theorem \ref{thm:PhiStability}, the aggregate system \eqref{eq:U(k+1)} is stable regardless of the choice of AIMD parameters.
    \item[3)] Since $\Phi$ in \eqref{eq:U(k+1)} is a Schur matrix, the $i$th AIMD admission rate converges to 
    \begin{equation}
        u_{i}^{\ast} = \frac{\alpha_{i}}{1-\beta_{i}}T^{\ast},
    \end{equation}
    where $T^{\ast} = \sum_{j=1}^{n}(\frac{\alpha_{i}}{2}\frac{1+\beta_{i}}{1-\beta_{i}})^{-1}\lambda$.
\end{itemize}

\section{Resource Allocation Control}\label{section:resourceAllocation}
Resource allocation in a queueing system pertains to a strategy ensuring that computing nodes provide incoming requests with adequate resources so that the number of queued requests is not increasing indefinitely as more requests are added to the system. Stabilizing the overall system, minimizing queueing and idle times, providing a trade-off between server utilization and application performance, and maximizing system throughput and output are essential objectives of resource allocation strategies in queueing systems. Here, we focus on stability as a fundamental qualitative property, which if not present, may make it impossible for a queueing scheme to achieve any other desirable objective. 
%
%


We follow a bottom-up approach for designing a decentralized resource allocation control strategy as follows. Let $(\alpha_{i},\;\beta_{i})$, $\gamma_{i}$ denote the AIMD parameters, and the service rate, respectively, associated with the $i$th node. Recall that, for $\tau\in[0,\;T(k)]$,
\begin{equation}
    u_{i}(\tau) = \beta_{i}u_{i}(k) + \alpha_{i}\tau, 
\end{equation}
is the rate at which requests are admitted to the $i$th node, while
\begin{equation}
    w_{i}(\tau) = w_{i}(k) + \beta_{i}u_{i}(k)\tau + \frac{\alpha_{i}}{2}\tau^{2} - \gamma_{i}(k) \tau,
\end{equation}
is the number of queued requests waiting in the $i$th node, during the $k$th cycle, respectively. We define by
\begin{align}
    y_{i}^{k}(\tau) &= w_{i}(\tau) + \gamma_{i}(k)\tau,\\  z_{i}^{k}(\tau) &= \gamma_{i}(k)\tau,
\end{align}
the total number of requests that have been admitted by time $\tau$, and the number of requests that can be served at most by time $\tau$, respectively. Let also $\hat{\gamma}_{i}(k)$ be the slope of a line segment starting from the origin tangent to parabola $y_{i}^{k}(\tau)$ (see $\mathbf{OA}$ in Fig. \ref{fig:tangent}). By letting $\gamma_{i}(k) = \hat{\gamma}_{i}(k)$, thus selecting $z_{i}^{k}$ as the line tangent to $y_{i}^{k}(\tau)$ at point $t_{z}^{k}\in[0,\;T(k)]$, as shown in Fig. \ref{fig:tangent}, we effectively guarantee that the maximum number of requests that can be served, during the $k$th cycle, never exceeds the actual number of admitted requests, avoiding, thus, node under-utilization. In Theorem \ref{thm:wi(k+1)} below, we also show that this resource allocation choice is stabilizing. Note also that if $\gamma_{i}(k)>\hat{\gamma}_{i}(k)$ (see red dashed line in Fig. \ref{fig:tangent}) there is always a nonzero interval that the queue of the $i$th node remains empty, i.e., resources are over-provisioned. Similarly, by letting $0<\gamma_{i}(k)<\hat{\gamma}_{i}(k)$ (see blue dashed line in Fig. \ref{fig:tangent}) there is no stability guarantee that the $i$th queue remains bounded. 

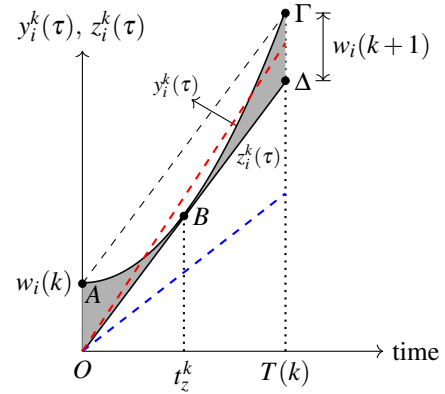
\begin{figure}[ht]
    \centering
\begin{tikzpicture}
  \draw[->] (0, 0) -- (4, 0) node[right] {time};
  \draw[->] (0, 0) -- (0, 4) node[above] {$y_{i}^{k}(\tau),\;
  z_{i}^{k}(\tau)$};
  \path (0, 0) node[below] {$O$};
  \path (0, .9) node[left] {$w_{i}(k)$};
  
  \path (2.7, 0) node[below] {$T(k)$};
  \draw[thick, dotted] (1.35, 1.8) -- (1.35, 0);
  \path (1.35, 0) node[below] {$t_{z}^{k}$};
  \draw[very thick, scale=1, domain=0:1.8, smooth, variable=\x, black, name path = A] plot ({1.5*\x}, {.9+ .2*(5.5556*\x*\x)});
  \draw[very thick, scale=1, domain=0:1.8, smooth, variable=\y, black, name path = B]  plot ({1.5*\y}, {.2*(10*\y)});
  \draw[thick, dotted] (2.7,0) -- (2.7,.9+ 3.6);
\tikzfillbetween[of=A and B]{black!30!white}; 
     \draw[fill] (0, .9) circle [radius=0.05cm];
     \draw[fill] (1.35, 1.8) circle [radius=0.05cm] node[right, yshift = .cm] {$B$};
    \draw[fill] (2.7,.9+ 3.6) circle [radius=0.05cm] node[right] {$\Gamma$};
    \draw[fill] (2.7,.0+ 3.6) circle [radius=0.05cm] node[right] {$\Delta$};
    \draw (3.1, .9+3.6) -- (3.3, .9+3.6);
    \draw (3.1, .+3.6) -- (3.3, .+3.6);
    \draw[->] (3.2, .9+3.6) -- (3.2, .+3.6);
    \draw[<-] (3.2, .9+3.6) -- (3.2, .9+3.6);
    \path (3.2, .45+3.6) node[right] {$w_{i}(k+1)$};
    \path (1.25, 3.25) node[above] {\footnotesize$y_{i}^{k}(\tau)$};
    \draw[->] (2.05, 3) -- (1.45, 3.35);
    \path (2.35, 2.3) node[above] {\footnotesize$z_{i}^{k}(\tau)$};
    \draw[dashed, name path = C] (0, .9) -- (2.7,.9+ 3.6);
    \path (-.08, 1.) node[anchor = north west] {$A$};
    \draw[thick, dashed,blue] (0,0) -- (2.7,-1.5+ 3.6);
    \draw[thick, dashed,red] (0,0) -- (2.7,.5+ 3.6);
\end{tikzpicture}
\caption{Arrivals and departures in the $i$th node during the $k$th cycle.}
    \label{fig:tangent}
\end{figure}

We now show how to obtain a closed formula for $\hat{\gamma}_{i}(k)$. We first find the intersection point $B$, as shown in Fig. \ref{fig:tangent}, where
\begin{align}
    y_{i}^{k}(t_{z}^{k}) & = z_{i}^{k}(t_{z}^{k}), \label{eq:functions_eq} \\ 
    \diff{y_{i}^{k}}{t_{z}^{k}} & = \diff{z_{i}^{k}}{t_{z}^{k}}. \label{eq:deriv_eq}
\end{align}
From \eqref{eq:deriv_eq}, we get 
\begin{equation}\label{eq:tzk1}
    t_{z}^{k} = \frac{\hat{\gamma}_{i}(k) - \beta_{i}u_{i}(k)}{\alpha_{i}},
\end{equation}
while, after a few calculations, using \eqref{eq:tzk1} in \eqref{eq:functions_eq}, we have
\begin{equation}\label{eq:hatgammaik}
    \hat{\gamma}_{i}(k) = \beta_{i}u_{i}(k) + \sqrt{2\alpha_{i}w_{i}(k)},
\end{equation}
which is a nonlinear, discrete-time state-feedback controller. Finally, using \eqref{eq:hatgammaik}, \eqref{eq:tzk1} becomes
\begin{equation}
    t_{z}^{k} = \sqrt{\frac{2w_{i}(k)}{\alpha_{i}}}.
\end{equation}
We are now in a position to state the second main result of our work, namely, the proposed resource allocation strategy along with its stability properties.

\begin{theorem}\label{thm:wi(k+1)}
Let
\begin{equation}\label{eq:wi(k+1)theorem}
    w_{i}(k+1) = w_{i}(k) +(\beta_{i}u_{i}(k) + \frac{\alpha_{i}}{2}T(k) - \gamma_{i}(k))T(k),
\end{equation}
with $ w_{i}(0) \geq 0$, be the queue dynamics of the $i$th node, where $\alpha_{i}$, $\beta_{i}$ are the AIMD parameters, $T(k)$ is the cycle period, and 
\begin{equation}\label{eq:feedback_gamma_i}
    \gamma_{i}(k) = \beta_{i}u_{i}(k) + \sqrt{2\alpha_{i}w_{i}(k)},
\end{equation}
is a feedback resource allocation policy. Then, the following are true.
\begin{itemize}
    \item[i.)] System \eqref{eq:wi(k+1)theorem}-\eqref{eq:feedback_gamma_i} is nonnegative for all $w_{i}(k)\geq 0$.
    \item[ii.)] The set $\mathcal{W}_i(k)=[0,\; \frac{a_i}{2}T(k)^2]$ is invariant with respect to system \eqref{eq:wi(k+1)theorem}-\eqref{eq:feedback_gamma_i}.
    \item[iii.)] For $w_{i,0}\notin\mathcal{W}_{i}(k)$, there is an integer $k^\star>0$ such that $w_{i}(k^\star) \in \mathcal{W}_{i}(k)$.
\end{itemize}
\end{theorem}
\begin{proof}
i.) In the proof, we denote $\diff{\phi}{x}$ by $\phi'(x)$. Substituting \eqref{eq:feedback_gamma_i} in \eqref{eq:wi(k+1)theorem}, we write
\begin{equation}
    w_{i}(k+1) = w_{i}(k) +\frac{\alpha_{i}}{2}T(k)^{2} - \sqrt{2\alpha_{i}w_{i}(k)}T(k),
\end{equation}
and we show that $f(w_{i}(k)) = w_{i}(k+1)$ is convex in $w_{i}(k)$. Indeed, $f(w_{i}(k))$ is convex with respect to $w_i(k)$ since it is a sum of the affine function $w_{i}(k) + \frac{\alpha_{i}}{2}T(k)^{2}$ and the convex function $- \sqrt{2\alpha_{i}w_{i}(k)}T(k)$. Note also that $f(w_{i}(k))$ is convex for all $T(k)\geq 0$. Since, $f(w_{i}(k))$ is continuously differentiable and convex for $w_{i}(k)\geq 0$, the unique minimizer is attained by setting $f'(w_{i}(k)) = 0$, which results in $1-\frac{\sqrt{2\alpha_{i}}T(k)}{2\sqrt{w_{i}(k)}} = 0$, or
\begin{equation*}
w_{i}^{*}(k) = \frac{\alpha_{i}T(k)^{2}}{2}
\end{equation*}
Taking into account that $f(w_{i}^{*}(k)) = 0$, it holds that $f(w_{i}(k))\geq 0$ $\forall\; w_{i}(k)\geq 0$, $T(k)\geq 0$.

ii.) The condition $w_{i}(k+1)\leq w_{i}(k)$ holds when $ w_{i}(k) + \frac{\alpha_{i}}{2}T(k)^{2} - \sqrt{2\alpha_{i}w_{i}(k)}T(k) \leq w_{i}(k)$, or when
\begin{equation*}
w_{i}(k) \geq \frac{\alpha_{i}}{8}T(k)^{2}.    
\end{equation*}
 Let $\hat{\mathcal{W}}_{i}(k) =\left\{w\in\R: w\geq \frac{\alpha_{i}}{8}T(k)^{2}\right\}$. Since $\hat{\mathcal{W}}_{i}(k)\cap \mathcal{W}_{i}(k) =  [\frac{\alpha_{i}}{8}T(k)^{2},\;\frac{\alpha_{i}}{2}T(k)^{2}]$, we need only to verify $w_{i}(k+1) \leq w_{i}(k)$ $\forall w_{i}(k)\in [0,\; \frac{\alpha_{i}}{8}T(k)^{2})$. Since  $f(w_{i}(k))$ is convex with minimum at $\frac{\alpha_{i}}{2}T(k)^{2}$ it follows that $f(0)\geq f(w_{i}(k))$ for any $ w_{i}(k)\in [0,\;\frac{\alpha_{i}}{2}T(k)^{2}]$. Since $f(0) = \frac{\alpha_{i}}{2}T(k)^{2}$, it holds that $0\leq f(w_{i}(k))\leq \frac{\alpha_{i}}{2}T(k)^{2}$ for all $ w_{i}(k)\in \mathcal{W}_{i}(k)$. This proves part $ii.)$.

iii.) Consider function $g(w_{i}(k)) = w_{i}(k) - f(w_{i}(k))$. Then, $g'(w_{i}(k)) = 1 - f'(w_{i}(k)) = \frac{\sqrt{2\alpha_{i}}T(k)}{2\sqrt{w_{i}(k)}}>0$ since $T(k)>0$ for all $k\geq 0$. Moreover, $g(\frac{\alpha_{i}}{2}T(k)^{2}) = \frac{\alpha_{i}}{2}T(k)^{2}$. Thus, $\forall k>0$, and $\forall w_{i}(k)\geq \frac{\alpha_{i}}{2}T(k)^{2}$, we have $g(w_{i}(k))\geq \frac{\alpha_{i}}{2}T(k)^{2}$, or $f(w_{i}(k))\leq w_{i}(k) - \frac{\alpha_{i}}{2}T(k)$. We now claim that for any $w_{i}(0) \geq \frac{\alpha_{i}}{2}T(0)^{2}$, $\exists$ $k_{i}^{\ast}$ such that $w_{i}(k_{i}^{\ast}) \leq \frac{\alpha_{i}}{2}T(k_{i}^{\ast})^{2}$. Indeed, $w_{i}(k_{i}^{\ast}) \leq w_{i}(0) - \sum_{j=0}^{k_{i}^{\ast}-1}\frac{\alpha_{i}}{2}T(j)$. To enforce the claim, we have $w_{i}(0) - \sum_{j=0}^{k_{i}^{\ast}-1}\frac{\alpha_{i}}{2}T(j) \leq \frac{\alpha_{i}}{2}T(k_{i}^{\ast})^{2}$, or $w_{i}(0) - \sum_{j=0}^{k_{i}^{\ast}-1}\frac{\alpha_{i}}{2}T(j) \leq w_{i}(0) - (\frac{\alpha_{i}}{2}\min_{j=0,\ldots,k_{i}^{\ast}-1}T(j)^{2})k_{i}^{\ast}\leq \frac{\alpha_{i}}{2}T(k_{i}^{\ast})^{2}$, thus
\begin{equation*}
    k_{i}^{\ast}\geq \Big\lceil\frac{w_{i}(0)-\frac{\alpha_{i}}{2}T(k_{i}^{\ast})}{\frac{\alpha_{i}}{2}\min_{j=0,\ldots,k_{i}^{\ast}-1}T(j)^{2}}\Big\rceil,
\end{equation*}
which can always be found. 

\end{proof}

It is worth highlighting some appealing characteristics of the proposed resource allocation scheme:
\begin{itemize}
    \item[1)] The closed-loop system \eqref{eq:wi(k+1)theorem}-\eqref{eq:feedback_gamma_i} under the resource allocation policy $\gamma_{i}(k) = \hat{\gamma}_{i}(k)$ is globally attracted to the sets $\mathcal{W}_i(k)=[0,\; \frac{a_i}{2}T(k)^2]$ from Theorem 3 in finite time.
    \item[2)] Implementation of \eqref{eq:feedback_gamma_i} is decentralised for any $i=1,..,n$ as only local information is required.
    \item[3)] The stability properties of $\gamma_{i}(k) = \hat{\gamma}_{i}(k)$ are independent of the particular tuning of the AIMD parameters $\alpha_{i}$, $\beta_{i}$, $i=1,\ldots,n$.
    \item[4)] The proposed resource allocation strategy \eqref{eq:feedback_gamma_i} is scalable irrespective of the total number of individual nodes.
\end{itemize}



\section{Queuing Time Calculation}\label{section:queueingTime}

We define the total queueing time associated with requests dispatched to the $i$th node as
\begin{equation}
    T_{i}(k) = T_{\delta_{i}}(k) + T_{w_{i}}(k),
\end{equation}
where $T_{\delta_{i}}(k)$ corresponds to queueing time in the batch queue, while $T_{w_{i}}(k)$ corresponds to queueing time in the $i$th node. We also define the following. The average admission rate associated with the $i$th node is defined as
\begin{equation}
    u_{i}^{\mathbf{av}}(k) = \frac{1}{T(k)}\int_{t_{k}}^{t_{k+1}} u_{i}(t)  \mathbf{d}t,
\end{equation}
where $T(k) = t_{k+1} - t_{k}$, and $u_{i}(t)$ is given by \eqref{eq:u_i(t)}. Solving the integral above yields
\begin{equation}\label{eq:uiav(k)}
    u_{i}^{\mathbf{av}}(k) = \beta_{i}u_{i}(k) + \frac{\alpha_{i}}{2} T(k).
\end{equation}
 In view of \eqref{eq:lambda} and \eqref{eq:uiav(k)}, we may write that
\begin{equation}\label{eq:AIMDcondition}
    \sum_{i=1}^{n} u_{i}^{\mathbf{av}}(k) = \lambda, \; \forall \; k\geq 0.
\end{equation}
 In view of \eqref{eq:AIMDcondition}, we can write that $u_{i}^{\mathbf{av}}(k)T(k)$ corresponds to the fraction of the total arrivals at the batch queue (namely, $\sum_{i=1}^{n}u_{i}^{\mathbf{av}}(k)T(k) = \lambda T(k)$) associated with the $i$th node. Using definition \eqref{eq:queueingTime}, and letting $\delta_{i}(\tau) = u_{i}^{\mathbf{av}}(k)\tau - \beta_{i}u_{i}(k)\tau - \frac{\alpha_{i}}{2}\tau^{2}$, with $0 \leq \tau \leq T(k)$, be the number of queued requests waiting in the batch queue before being dispatched to the $i$th node, we may write that $T_{\delta_{i}}(k) = \frac{\int_{0}^{T(k)}\delta_{i}(\tau)\mathbf{d}\tau}{u_{i}^{\mathbf{av}}(k)T(k)}$. Interestingly enough, integral $\int_{0}^{T(k)}\delta_{i}(\tau)\mathbf{d}\tau$ is identically equal to the unshaded area $AB\Gamma A$, in Fig. \ref{fig:tangent}. Similarly, using definition \eqref{eq:queueingTime}, we may write that $T_{w_{i}}(k) = \frac{\int_{0}^{T(k)}w_{i}(\tau)\mathbf{d}\tau}{\beta_{i}u_{i}(k)T(k) + \frac{1}{2}\alpha_{i}T(k)^{2}}$, where integral $\int_{0}^{T(k)}w_{i}(\tau)\mathbf{d}\tau$ is identically equal to the shaded area in Fig. \ref{fig:tangent}, and $\beta_{i}u_{i}(k)T(k) + \frac{1}{2}\alpha_{i}T(k)^{2} = u_{i}^{\mathbf{av}}(k)T(k)$ due to \eqref{eq:uiav(k)}. Adding the two aforementioned areas and dividing by $u_{i}^{\mathbf{av}}(k)T(k)$ clearly yields the queueing time associated with the $i$th node. In other words, $T_{i}(k)$ is equal to the area of the trapezium $A\Gamma\Delta O$ divided by $u_{i}^{\mathbf{av}}(k)T(k)$, i.e., $ T_{i}^{\mathbf{tot}}(k) = \frac{(w_{i}(k)+w_{i}(k+1))\frac{T(k)}{2}}{u_{i}^{\mathbf{av}}(k)T(k)}$
or
\begin{equation}\label{eq:T_i(k)}
    T_{i}(k) = \frac{w_{i}(k)+w_{i}(k+1)}{2u_{i}^{\mathbf{av}}(k)}.
\end{equation}

\begin{remark}
The total queueing time of node $i$ can be defined by means of local information without any information pertinent to the batch queue. Also, by defining $w_{i}^{\mathbf{av}}(k) = (w_{i}(k)+w_{i}(k+1))/2$ as the average number of queued requests over the $k$th cycle, \eqref{eq:T_i(k)} becomes
\begin{equation}
    T_{i}(k) = \frac{w_{i}^{\mathbf{av}}(k)}{u_{i}^{\mathbf{av}}(k)},
\end{equation}
which is clearly consistent with \textit{Little's Law}.
\end{remark}

\section{Numerical Example}\label{section:example}
We consider a flow of requests with constant flow rate, $\lambda$ [req/sec], entering a system of four computing nodes. The $i$th node is associated with a FCFS queue the length of which is denoted by $w_{i}$, $i=1,\ldots,4$. By tuning parameters $\alpha_{i}$, $\beta_{i}$ independently, the $i$th node admits requests according to AIMD control policy $u_{i}$ given in \eqref{eq:u_ik+1}, and alters its service rate via discrete nonlinear feedback controller $\gamma_{i}$ defined in \eqref{eq:feedback_gamma_i}. Simulation results are presented in Fig. \ref{fig:1}-\ref{fig:6} for the setup parameters shown in Table \ref{table:1}. 

\begin{table}[h]
\caption{Simulation parameters}
\label{table:1}
\begin{center}
\begin{tabular}{|c||c||c||c||c||c|}
\hline
$\lambda$ & $\alpha_{i}$ & $\beta_{i}$ & $u_{i}(0)$ & $w_{i}(0)$ & $i$\\
\hline
100 & $5i$  & $0.5$  & $(i-1)5$ & $(2i-1) 7.5$ & $\{1,2,3,4\}$ \\
\hline
\end{tabular}
\end{center}
\end{table}
\begin{table}[h]
\caption{Invariant sets for $k\geq 15$}
\label{table:2}
\begin{center}
\begin{tabular}{|c||c||c||c|}
\hline
$\mathcal{W}_{1}(k)$ & $\mathcal{W}_{2}(k)$ & $\mathcal{W}_{3}(k)$ & $\mathcal{W}_{4}(k)$\\
\hline
$[0,\;4.44]$ & $[0,\;8.88]$  & $[0,\;13.33]$  & $[0,\;17.77]$ \\
\hline
\end{tabular}
\end{center}
\end{table}
\begin{figure}
    \centering
    \scalebox{.4}{
    \includegraphics{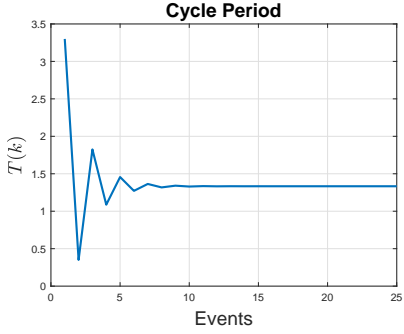}}
    \caption{Cycle Period}
    \label{fig:1}
\end{figure}
\begin{figure}
    \centering
    \scalebox{.4}{
    \includegraphics{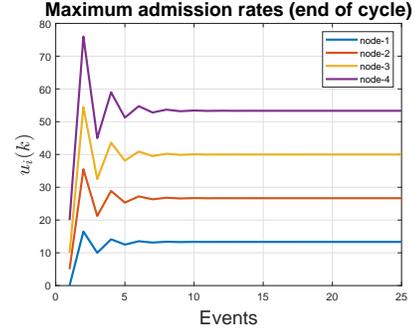}}
    \caption{Maximum admission rates before each clearance event.}
    \label{fig:2}
\end{figure}
\begin{figure}
    \centering
    \scalebox{.4}{
    \includegraphics{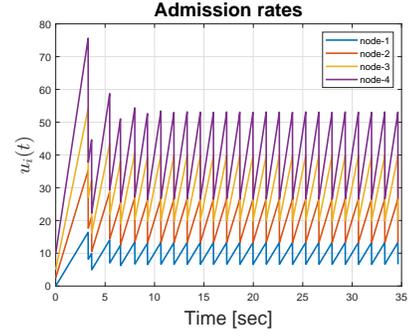}}
    \caption{AIMD admission control.}
    \label{fig:3}
\end{figure}
\begin{figure}
    \centering
    \scalebox{.4}{
    \includegraphics{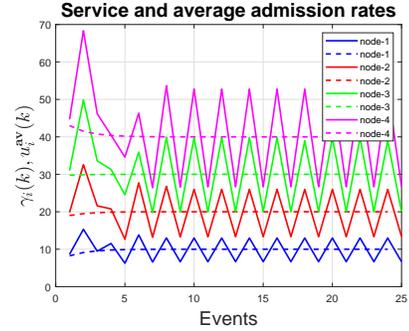}}
    \caption{Resource allocation and average admission control. Solid lines: $\gamma_{i}(k)$, dashed lines: $u_{i}^{\mathbf{av}}(k)$.}
    \label{fig:4}
\end{figure}
\begin{figure}
    \centering
    \scalebox{.4}{
    \includegraphics{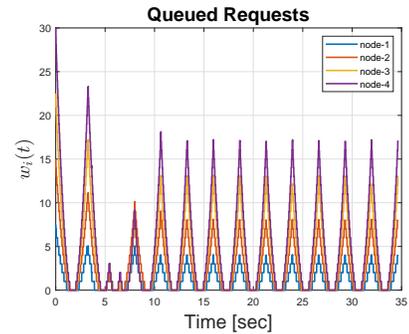}}
    \caption{Queued requests.}
    \label{fig:5}
\end{figure}
\begin{figure}
    \centering
    \scalebox{.4}{
    \includegraphics{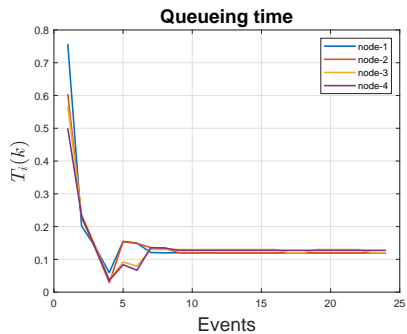}}
    \caption{Queueing time.}
    \label{fig:6}
\end{figure}

As can be seen from Fig. \ref{fig:1}, the cycle period converges as expected to $T^{\ast} = 1.33$ [sec]. Viewing Fig. \ref{fig:2}, admission rates $u_{i}(k)$, $i=1,\ldots,4$, also converge to $u_{i}^{\ast} = \frac{\alpha_{i}}{1-\beta_{i}}T^{\ast}$, $i=1,\ldots,4$, verifying the validity of Theorem \ref{thm:PhiStability}. Fig. \ref{fig:3} illustrates typical AIMD behaviour with convergence occurring after approximately  $10$ events. This convergence rate is related to the particular choice of AIMD parameters. For example, faster convergence is expected if growth rates $\alpha_{i}$, $i=1,\ldots,4$, are selected more aggressively. Service rates depicted in Fig. \ref{fig:4} are calculated according to resource allocation law \eqref{eq:feedback_gamma_i}. From the figure, it is evident that the service rate mean value of each node is heavily related to the corresponding average AIMD admission rate. Queue profiles are shown in Fig. \ref{fig:5}, where it is evident that queues are bounded highlighting the stability properties of Theorem \ref{thm:wi(k+1)}. Invariant sets $\mathcal{W}_{i}(k)$ for $T(k) = T^{\ast}$ are given in Table \ref{table:2}. Overall, we note that under the proposed scheduling and resource allocation strategy, stable operation is guaranteed for all computing nodes regardless of the tuning of individual AIMD parameters. We refer interested readers to \cite{githubcdc21} for further simulation scenarios with arbitrary number of nodes. Therein, a script for a random arrival process with exponentially distributed inter-arrival times is also available.


\section{Conclusion}\label{section:conclusion}

We study the problem of simultaneous scheduling and resource allocation of a deterministic flow of requests entering a system of computing nodes which is represented as a multi-queue scheme. Inspired by the well-established AIMD algorithm, we present a new admission control policy for general scheduling problems. Using an interesting property of rank-one perturbations of symmetric matrices, we provide stability guarantees, independent of the overall system dimension and the AIMD tuning. Following a bottom-up approach, we then propose a resource allocation strategy defined as a decentralized nonlinear feedback controller which is globally stabilizing. This effectively guarantees that individual queues are bounded converging in finite time to a well-defined interval. Finally, we associated these properties with Quality of Service specifications, by calculating the local queueing time via a simple formula consistent with Little's Law. Our method is simple, scalable, and locally configurable. It is worth noting however that further effort is required to formally address two additional challenges, namely, non deterministic workload and the presence of resource constraints. This is the subject of our immediate future research efforts. 

\bibliographystyle{IEEEtran} 
\balance
\bibliography{IEEEabrv,biblio}

\begin{thebibliography}{10}
\providecommand{\url}[1]{#1}
\csname url@rmstyle\endcsname
\providecommand{\newblock}{\relax}
\providecommand{\bibinfo}[2]{#2}
\providecommand\BIBentrySTDinterwordspacing{\spaceskip=0pt\relax}
\providecommand\BIBentryALTinterwordstretchfactor{4}
\providecommand\BIBentryALTinterwordspacing{\spaceskip=\fontdimen2\font plus
\BIBentryALTinterwordstretchfactor\fontdimen3\font minus
  \fontdimen4\font\relax}
\providecommand\BIBforeignlanguage[2]{{%
\expandafter\ifx\csname l@#1\endcsname\relax
\typeout{** WARNING: IEEEtran.bst: No hyphenation pattern has been}%
\typeout{** loaded for the language `#1'. Using the pattern for}%
\typeout{** the default language instead.}%
\else
\language=\csname l@#1\endcsname
\fi
#2}}

\bibitem{Kehoe2015}
B.~Kehoe, S.~Patil, P.~Abbeel, and K.~Goldberg, ``{A Survey of Research on
  Cloud Robotics and Automation},'' \emph{IEEE Transactions on Automation
  Science and Engineering}, vol.~12, no.~2, pp. 398--409, 2015.

\bibitem{Mach2017}
P.~Mach and Z.~Becvar, ``{Mobile Edge Computing: A Survey on Architecture and
  Computation Offloading},'' \emph{IEEE Communications Surveys and Tutorials},
  vol.~19, no.~3, pp. 1628--1656, 2017.

\bibitem{Abbas2018}
N.~Abbas, Y.~Zhang, A.~Taherkordi, and T.~Skeie, ``{Mobile Edge Computing: A
  Survey},'' \emph{IEEE Internet of Things Journal}, vol.~5, no.~1, pp.
  450--465, 2018.

\bibitem{Karamanolis2005}
C.~Karamanolis, M.~Karlsson, and X.~Zhu, ``{Designing controllable computer
  systems},'' in \emph{Proceedings of 10th Workshop on Hot Topics in Operating
  Systems}.\hskip 1em plus 0.5em minus 0.4em\relax Berkeley, CA, USA: USENIX
  Association, 2005, pp. 9--15.

\bibitem{Wang2005}
Z.~Wang, X.~Zhu, and S.~Singhal, ``{Utilization and SLO-based control for
  dynamic sizing of resource partitions},'' in \emph{16th IFIP/IEEE Ambient
  Networks international conference on Distributed Systems: operations and
  Management}.\hskip 1em plus 0.5em minus 0.4em\relax Springer Verlag, 2005,
  pp. 133--144.

\bibitem{Dechouniotis2015}
D.~Dechouniotis, N.~Leontiou, N.~Athanasopoulos, A.~Christakidis, and
  S.~Denazis, ``{A control-theoretic approach towards joint admission control
  and resource allocation of cloud computing services},'' \emph{International
  Journal of Network Management}, vol.~25, no.~3, pp. 159--180, 2015.

\bibitem{Avgeris2019}
M.~Avgeris, D.~Dechouniotis, N.~Athanasopoulos, and S.~Papavassiliou,
  ``{Adaptive resource allocation for computation offloading: A
  control-theoretic approach},'' \emph{ACM Transactions on Internet
  Technology}, vol.~19, no.~2, pp. 1--20, 2019.

\bibitem{Maggio2010}
M.~Maggio, H.~Hoffmann, M.~D. Santambrogio, A.~Agarwal, and A.~Leva,
  ``{Controlling software applications via resource allocation within the
  Heartbeats framework},'' in \emph{Proceedings of the IEEE Conference on
  Decision and Control}, 2010, pp. 3736--3741.

\bibitem{Kalyvianaki2014}
E.~Kalyvianaki, T.~Charalambous, and S.~Hand, ``{Adaptive resource provisioning
  for virtualized servers using kalman filters},'' \emph{ACM Transactions on
  Autonomous and Adaptive Systems}, vol.~9, no.~2, pp. 1--35, 2014.

\bibitem{Makridis2018}
E.~Makridis, K.~Deliparaschos, E.~Kalyvianaki, A.~Zolotas, and T.~Charalambous,
  ``{Robust dynamic CPU resource provisioning in virtualized servers},''
  \emph{arXiv}, no. January, 2018.

\bibitem{Chiu1989}
D.~M. Chiu and R.~Jain, ``{Analysis of the increase and decrease algorithms for
  congestion avoidance in computer networks},'' \emph{Computer Networks and
  ISDN Systems}, vol.~17, no.~1, pp. 1--14, 1989.

\bibitem{Corless2016}
M.~Corless, C.~King, R.~Shorten, and F.~Wirth, \emph{{AIMD Dynamics and
  Distributed Resource Allocation}}.\hskip 1em plus 0.5em minus 0.4em\relax
  Society for Industrial and Applied Mathematics, 2016.

\bibitem{Berman2004}
A.~Berman, R.~Shorten, and D.~Leith, ``{Positive matrices associated with
  synchronised communication networks},'' \emph{Linear Algebra and Its
  Applications}, vol. 393, no. 1-3, pp. 47--54, 2004.

\bibitem{Shorten2005}
R.~N. Shorten, D.~J. Leith, J.~Foy, and R.~Kilduff, ``{Analysis and design of
  AIMD congestion control algorithms in communication networks},''
  \emph{Automatica}, vol.~41, no.~4, pp. 725--730, 2005.

\bibitem{Golub1973}
G.~Golub, ``{Some Modified Matrix Eigenvalue Problems},'' \emph{SIAM Review},
  vol.~15, no.~2, pp. 318--334, 1973.

\bibitem{Bunch1978}
J.~R. Bunch, C.~P. Nielsen, and D.~C. Sorensen, ``{Rank-one modification of the
  symmetric eigenproblem},'' \emph{Numerische Mathematik}, vol.~31, no.~1, pp.
  31--48, 1978.

\bibitem{Shorten2006}
R.~Shorten, F.~Wirth, and D.~Leith, ``{A positive systems model of TCP-like
  congestion control: Asymptotic results},'' \emph{IEEE/ACM Transactions on
  Networking}, vol.~14, no.~3, pp. 616--629, 2006.

\bibitem{Blanchini2015}
F.~Blanchini and S.~Miani, \emph{{Set-Theoretic Methods in Control}}, ser.
  Systems \& Control: Foundations \& Applications.\hskip 1em plus 0.5em minus
  0.4em\relax Birkh{\"{a}}user, 2015.

\bibitem{Maguluri2013}
S.~T. Maguluri and R.~Srikant, ``{Scheduling jobs with unknown duration in
  clouds},'' in \emph{Proceedings - IEEE INFOCOM}, 2013, pp. 1887--1895.

\bibitem{Maguluri2014}
S.~T. Maguluri, R.~Srikant, and L.~Ying, ``{Heavy traffic optimal resource
  allocation algorithms for cloud computing clusters},'' \emph{Performance
  Evaluation}, vol.~81, pp. 20--39, 2014.

\bibitem{Kleinrock1975}
L.~Kleinrock, \emph{{QUEUEING SYSTEMS, Volume I: Theory}}.\hskip 1em plus 0.5em
  minus 0.4em\relax Wiley-Interscience, 1975.

\bibitem{Cassandras2008}
C.~G. Cassandras and S.~Lafortune, \emph{{Introduction to discrete event
  systems}}.\hskip 1em plus 0.5em minus 0.4em\relax Springer US, 2008.

\bibitem{githubcdc21}
\BIBentryALTinterwordspacing
 [Online]. Available: \url{https://github.com/lefterisvl83/cdc21}
\BIBentrySTDinterwordspacing

\end{thebibliography}

\end{document}